\documentclass{basm}

\newcommand{\case}[1]{\textit{Case $\mathit{#1}$.}\ }

 \author{R.R. Kamalian}

\title{On cyclically-interval edge colorings of trees}

\RunningHead{R.R. Kamalian}%
            {On cyclically-interval edge colorings of trees}



\Abstract{For an undirected, simple, finite, connected graph $G$, we
denote by $V(G)$ and $E(G)$ the sets of its vertices and edges,
respectively. A function $\varphi:E(G)\rightarrow\{1,2,\ldots,t\}$
is called a proper edge $t$-coloring of a graph $G$ if adjacent
edges are colored differently and each of $t$ colors is used. An
arbitrary nonempty subset of consecutive integers is called an
interval. If $\varphi$ is a proper edge $t$-coloring of a graph $G$
and $x\in V(G)$, then $S_G(x,\varphi)$ denotes the set of colors of
edges of $G$ which are incident with $x$. A proper edge $t$-coloring
$\varphi$ of a graph $G$ is called a cyclically-interval
$t$-coloring if for any $x\in V(G)$ at least one of the following
two conditions holds: a) $S_G(x,\varphi)$ is an interval, b)
$\{1,2,\ldots,t\}\setminus S_G(x,\varphi)$ is an interval. For any
$t\in \mathbb{N}$, let $\mathfrak{M}_t$ be the set of graphs for
which there exists a cyclically-interval $t$-coloring, and let
$$\mathfrak{M}\equiv\bigcup_{t\geq1}\mathfrak{M}_t.$$
For an arbitrary tree $G$, it is proved that $G\in\mathfrak{M}$ and
all possible values of $t$ are found for which
$G\in\mathfrak{M}_t.$}

\CopyRight { R.R. Kamalian, 2012}

\Address{{\sc R.R. Kamalian}\\
{Institute for Informatics and Automation Problems National Academy
of Sciences of RA, 0014 Yerevan, Republic of Armenia}\\
E-mail: \emph{rrkamalian@yahoo.com}}
\medskip

 \Received { \ August 30, 2011}

\begin{document}

\maketitle

\section{Introduction}
We consider undirected, simple, finite, and connected graphs. For a
graph $G$ we denote by $V(G)$ and $E(G)$ the sets of its vertices
and edges, respectively. The set of edges of $G$ incident with a
vertex $x\in V(G)$ is denoted by $J_G(x)$. The set of vertices of
$G$ adjacent to a vertex $x\in V(G)$ is denoted by $I_G(x)$. For any
$x\in V(G)$, $d_{G}(x)$ denotes the degree of the vertex $x$ in $G$.
For a graph $G$, we denote by $\Delta (G)$ and $\chi ^{\prime }(G)$
the maximum degree of a vertex of $G$ and the chromatic index of $G$
\cite{Vizing2}, respectively. The distance in a graph $G$ between
its vertices $x\in V(G)$ and $y\in V(G)$ is denoted by
$\rho_G(x,y)$. For any vertex $x_0\in V(G)$ and an arbitrary subset
$V_0$ of the set $V(G)$, we define the distance $\rho_G(x_0,V_0)$ in
a graph $G$ between $x_0$ and $V_0$ as follows:
$$\rho_G(x_0,V_0)\equiv\min_{z\in V_0}\rho_G(x_0,z)$$

For any integer $n\geq 3$, we denote by $C_n$ a simple cycle with
$n$ vertices. The terms and concepts that we do not define can be
found in \cite{West1}.

For an arbitrary finite set $A$, we denote by $|A|$ the number of
elements of $A$. The set of positive integers is denoted by
$\mathbb{N}$. An arbitrary nonempty subset of consecutive integers
is called an interval. An interval with the minimum element $p$ and
the maximum element $q$ is denoted by $[p,q]$. An interval $D$ is
called a $h$-interval if $|D|=h$.

For any $t\in \mathbb{N}$ and arbitrary integers $i_{1},i_{2}$
satisfying the conditions $i_{1}\in[1,t]$, $i_{2}\in[1,t]$, we
define \cite{Shved1_11} the sets $intcyc_{1}((i_{1},i_{2}),t),$
$intcyc_{1}[(i_{1},i_{2}),t],$ $intcyc_{2}((i_{1},i_{2}),t),$
$intcyc_{2}[(i_{1},i_{2}),t]$ and the number $dif((i_{1},i_{2}),t)$
as follows:
$$intcyc_{1}[(i_{1},i_{2}),t]\equiv [\min
\{i_{1},i_{2}\},\max\{i_{1},i_{2}\}],$$
$$intcyc_{1}((i_{1},i_{2}),t)\equiv
intcyc_{1}[(i_{1},i_{2}),t]\backslash (\{i_{1}\}\cup \{i_{2}\}),$$
$$intcyc_{2}((i_{1},i_{2}),t)\equiv [1,t]\backslash
intcyc_{1}[(i_{1},i_{2}),t],$$
$$intcyc_{2}[(i_{1},i_{2}),t]\equiv [1,t]\backslash
intcyc_{1}((i_{1},i_{2}),t),$$
$$dif((i_{1},i_{2}),t)\equiv \min\{\left\vert
intcyc_{1}[(i_{1},i_{2}),t]\right\vert ,\left\vert
intcyc_{2}[(i_{1},i_{2}),t]\right\vert\}-1.$$

If $t\in \mathbb{N}$ and $Q$ is a non-empty subset of the set
$\mathbb{N}$, then $Q$ is called a $t$-cyclic interval if there
exist integers $i_{1},i_{2},j_{0}$ satisfying the conditions
$i_{1}\in[1,t]$, $i_{2}\in[1,t]$, $ j_{0}\in \{1,2\}$,
$Q=intcyc_{j_{0}}[(i_{1},i_{2}),t]$.

A function $\varphi:E(G)\rightarrow [1,t]$ is called a proper edge
$t$-coloring of a graph $G$ if adjacent edges are colored
differently and each of $t$ colors is used.

If $\varphi$ is a proper edge $t$-coloring of a graph $G$ and
$E_{0}\subseteq E(G)$, then $\varphi [E_{0}]\equiv \{\varphi
(e)/e\in E_{0}\}$.

A proper edge $t$-coloring $\varphi $ of a graph $G$ is called an
interval $t$-coloring of $G$ \cite{Diss3, Oranj4, Canada} if for any
$x\in V(G)$, the set $\varphi[J_G(x)]$ is a $d_{G}(x)$-interval. For
any $t\in \mathbb{N}$, we denote by $\mathfrak{N}_t$ the set of
graphs for which there exists an interval $t$-coloring. Let us also
define the set $\mathfrak{N}$ of all interval colorable graphs:
$$\mathfrak{N}\equiv \bigcup_{t\geq1}\mathfrak{N}_t.$$
For any $G\in \mathfrak{N}$, we denote by $w_{int}(G)$ and
$W_{int}(G)$ the minimum and the maximum possible value of $t$,
respectively, for which $G\in \mathfrak{N}_t$. For a graph $G$, let
us set $\theta(G)\equiv\{t\in \mathbb{N}/G\in\mathfrak{N}_t\}$.

The problem of deciding whether a regular graph $G$ belongs to the
set $\mathfrak{N}$ is $NP$-complete \cite{Oranj4, Canada, Diss3}.
Nevertheless, for graphs $G$ of some classes the relation
$G\in\mathfrak{N}$ was proved and investigations of the set
$\theta(G)$ were fulfilled \cite{Oranj4, Canada, Diss3, Preprint5,
Petros_CSIT, Petros_DMath}. The concept of interval colorability of
a graph represents an especially high interest for a bipartite
graph, because in this case it can be used for mathematical
modelling of timetable problems with compactness requirements (i.e.
the lectures of each teacher and each group must be scheduled at
consecutive periods) \cite{Asratian_Diss, Cambridge, Diss3,
Petros_Akob}. Unfortunately, for an arbitrary bipartite graph $G$
the problem keeps the complexity of a general case \cite{Sev, Giaro,
Asratian_Cass2}. Some positive results were obtained for
\textquotedblleft small\textquotedblright
 bipartite graphs \cite{Giaro_Diss, Giaro_Kubale, Kubale}, for bipartite graphs with the
\textquotedblleft small\textquotedblright maximum degree of a vertex
\cite{Hansen_Dip, Giaro, Petros_Vest}, and for biregular bipartite
graphs \cite{Asratian_Cass1, Asratian_Cass2, Asratian_Cass3,
Asratian_Cass4, Asratian_Cass5, Cass, Hansen_Dip, Jensen_Toft,
Kamalian_Mir, Hans_Lot, Yang, Pyatkin}. Very interesting approaches
for biregular bipartite graphs were developed in \cite{Pyatkin,
Asratian_Cass5, Cass}. The examples of interval non-colorable
bipartite graphs were given in \cite{Cambridge, Giaro_Kubale,
Jensen_Toft, Sev}.

\begin{rem}
It is not difficult to see that for any integer $k\geq 2$,
$C_{2k}\in\mathfrak{N}$ and $\theta(C_{2k})=[2,k+1]$.
\end{rem}

A proper edge $t$-coloring $\varphi $ of a graph $G$ is called a
cyclically-interval $t$-coloring of $G$ if for any $x\in V(G)$, the
set $\varphi [J_G(x)]$ is a $t$-cyclic interval. For any $t\in
\mathbb{N}$, we denote by $\mathfrak{M}_t$ the set of graphs for
which there exists a cyclically-interval $t$-coloring. Let us also
define the set $\mathfrak{M}$ of all cyclically-interval colorable
graphs:
$$\mathfrak{M}\equiv\bigcup_{t\geq1}\mathfrak{M}_t.$$
For any $G\in \mathfrak{M}$, we denote by $w_{cyc}(G)$ and
$W_{cyc}(G)$ the minimum and the maximum possible value of $t$,
respectively, for which $G\in \mathfrak{M}_t.$ For a graph $G$, let
us set $\Theta(G)\equiv\{t\in \mathbb{N}/G\in\mathfrak{M}_t\}$.

\begin{rem}
The concept of cyclically-interval colorability of a graph
generalizes that of interval colorability. Clearly, for an arbitrary
graph $G\in\mathfrak{N}$, and for any $t\in\theta(G)$, an arbitrary
interval $t$-coloring of the graph $G$ is also a cyclically-interval
$t$-coloring of $G$, therefore, for any $t\in \mathbb{N}$,
$\mathfrak{N}_t\subseteq \mathfrak{M}_t$.
$\mathfrak{N}_2=\mathfrak{M}_2$. For any integer $t\geq3$,
$\mathfrak{N}_t\subset \mathfrak{M}_t$ (it is enough to consider the
simple cycle $C_t$). $\mathfrak{N}\subset\mathfrak{M}$ (it is enough
to consider the simple cycle $C_3$). For an arbitrary graph $G$,
$\theta(G)\subseteq\Theta(G)$.
\end{rem}

\begin{rem}
For any $G\in\mathfrak{N}$, the following inequality is true:
$$\Delta (G)\leq \chi ^{\prime }(G)\leq w_{cyc}(G)\leq
w_{int}(G)\leq W_{int}(G)\leq W_{cyc}(G)\leq \left\vert
E(G)\right\vert.$$
\end{rem}

\begin{rem}
It is not difficult to note that there exist examples $G_1$ and
$G_2$ of graphs from $\mathfrak{N}$ for which
$w_{cyc}(G_1)<w_{int}(G_1), W_{int}(G_2)<W_{cyc}(G_2)$. Let us set
$G_1=K_{3,2}$ and $G_2=K_{2,2}$. In this case, evidently,
$w_{cyc}(G_1)=3, w_{int}(G_1)=4$ \cite{Preprint5}, $W_{int}(G_2)=3$
\cite{Preprint5}, $W_{cyc}(G_2)=4$.
\end{rem}

The problem of cyclically-interval colorability of a graph
completely investigated as yet only for simple cycles \cite{Csit10,
Kamalian} and trees \cite{Shved1_11}. Some interesting results on
this and related topics were obtained in \cite{DeWerra6, DeWerra7,
Barth8, Daus9}.

For a tree $H$ with $V(H)=\{b_{1},...,b_{p}\},$ $p\geq 1$, we denote
by $P(b_{i},b_{j})$ the simple path connecting the vertices $b_{i}$
and $b_{j}$, $1\leq i\leq p$, $1\leq j\leq p$. The sets of vertices
and edges of the path $P(b_{i},b_{j})$ are denoted by
$VP(b_{i},b_{j})$ and $EP(b_{i},b_{j})$, respectively, $1\leq i\leq
p$, $1\leq j\leq p$.

Let us also define:
$$intVP(b_{i},b_{j})\equiv VP(b_{i},b_{j})\backslash(\{b_{i}\}\cup
\{b_{j}\});$$
$$\tilde{V}P(b_{i},b_{j})\equiv VP(b_{i},b_{j})\cup
\Bigg(\bigcup\limits_{x\in intVP(b_{i},b_{j})}I_H(x)\Bigg);$$
$$TP(b_{i},b_{j})\equiv \left\{
\begin{array}{lc}
\bigcup\limits_{x\in intVP(b_{i},b_{j})}J_H(x), & \textrm{if
$intVP(b_{i},b_{j})\neq\varnothing$} \\
EP(b_{i},b_{j}), & \textrm{if $intVP(b_{i},b_{j})=\varnothing$;}
\end{array}
\right.$$
$$1\leq i\leq p, 1\leq j\leq p.$$

Assume:
$$ M(H)\equiv\max\Big\{\big| TP(b_{i},b_{j})\big|/1\leq i\leq p, 1\leq j\leq p\Big\}.$$

In \cite{Preprint5} the following result was obtained.

\begin{thm} \cite{Preprint5}\label{Theorem1}
Let $H$ be an arbitrary tree.Then
\begin{enumerate}
\item $H\in \mathfrak{N}$,
\item $w_{int}(H)=\Delta (H),$
\item $W_{int}(H)=M(H),$
\item $\theta(H)=[\Delta(H),M(H)]$.
\end{enumerate}
\end{thm}

\begin{cor}\label{Cor1}
For any tree $H$, $H\in \mathfrak{M}$, $w_{cyc}(H)=\Delta(H)$,
$W_{cyc}(H)\geq M(H)$, $[\Delta(H),M(H)]\subseteq\Theta(H)$.
\end{cor}

In this paper, for any tree $H$, we show that $W_{cyc}(H)=M(H)$ and
$\Theta(H)=[\Delta(H), M(H)]$.

\section{Results}
\begin{lem}\label{Lemma1}
If $Q_{1},...,Q_{n}$ ($n\geq 2$) are $t$-cyclic intervals, and for
any $j\in[1,n-1]$, $Q_{j}\cap Q_{j+1}\neq \varnothing,$ then
$\bigcup\limits_{i=1}^{n}Q_{i}$ is a $t$-cyclic interval.
\end{lem}
\textit{Proof} can be easily accomplished by induction on $n$.

\begin{lem}\label{Lemma2}
Let $\alpha$ be a cyclically-interval $t$-coloring of a graph $G$,
and $P_{0}=(x_{0},e_{1},x_{1},...,x_{k-1},e_{k},x_{k})$ be a simple
path connecting a vertex $x_{0}\in V(G)$ with a vertex $x_{k}\in
V(G)$, $k\geq 2$. Then
$\alpha\Bigg[\bigcup\limits_{i=1}^{k-1}J_G(x_{i})\Bigg]$ is a
$t$-cyclic interval.
\end{lem}

\begin{proof}
If $k=2$, then the statement follows from the definition of the
cyclically-interval $t$-coloring. Now assume that $k\geq 3$. It is
clear that the sets $\alpha [J_G(x_{1})],...,\alpha[J_G(x_{k-1})]$
are $t$-cyclic intervals with
$$\alpha [J_G(x_{j})]\cap \alpha [J_G(x_{j+1})]\neq
\varnothing \textrm{ for any } j\in[1,k-2].$$

Lemma \ref{Lemma1} implies that
$\alpha\Bigg[\bigcup\limits_{i=1}^{k-1}J_G(x_{i})\Bigg]$ is a
$t$-cyclic interval.
\end{proof}

\begin{lem}\label{Lemma3}
Let $\alpha $ be a cyclically-interval $t$-coloring of a graph $G$,
and $P_{0}=(x_{0},e_{1},x_{1},...,x_{k-1},e_{k},x_{k})$ be a simple
path connecting a vertex $x_{0}\in V(G)$ with a vertex $x_{k}\in
V(G)$, $k\geq 2$. Then at least one of the following statements is
true:
\begin{enumerate}
\item $intcyc_{1}((\alpha (e_{1}),\alpha (e_{k})),t)\subseteq \alpha\Bigg[\bigcup\limits_{i=1}^{k-1}J_G(x_{i})\Bigg],$
\item $intcyc_{2}((\alpha (e_{1}),\alpha (e_{k})),t)\subseteq \alpha\Bigg[\bigcup\limits_{i=1}^{k-1}J_G(x_{i})\Bigg]$.
\end{enumerate}
\end{lem}

\begin{proof}
Without loss of generality we may assume that $dif((\alpha
(e_{1}),\alpha (e_{k})),t)\geq 2$.

Let us assume that none of the statements 1) and 2) is true. Then
there are $\tau _{1}$, $\tau _{2}$ such that
$$\tau _{1}\in intcyc_{1}((\alpha (e_{1}),\alpha (e_{k})),t),
\tau_{1}\not\in
\alpha\Bigg[\bigcup\limits_{i=1}^{k-1}J_G(x_{i})\Bigg],$$

$$\tau _{2}\in intcyc_{2}((\alpha (e_{1}),\alpha (e_{k})),t), \tau
_{2}\not\in
\alpha\Bigg[\bigcup\limits_{i=1}^{k-1}J_G(x_{i})\Bigg],$$ therefore
$\left\{ \tau _{1},\tau _{2}\right\} \cap
\alpha\Bigg[\bigcup\limits_{i=1}^{k-1}J_G(x_{i})\Bigg]=\varnothing
$.

Lemma \ref{Lemma2} implies that
$\alpha\Bigg[\bigcup\limits_{i=1}^{k-1}J_G(x_{i})\Bigg]$ is a
$t$-cyclic interval with
$$\left\{ \alpha (e_{1}),\alpha (e_{k})\right\}
\subseteq \alpha\Bigg[\bigcup\limits_{i=1}^{k-1}J_G(x_{i})\Bigg].$$

It is not hard to see that the relations
$$\left\{ \alpha(e_{1}),\alpha (e_{k})\right\} \subseteq \alpha\Bigg[\bigcup\limits_{i=1}^{k-1}J_G(x_{i})\Bigg]\textrm{ and }\left\{ \tau
_{1},\tau _{2}\right\} \cap
\alpha\Bigg[\bigcup\limits_{i=1}^{k-1}J_G(x_{i})\Bigg]=\varnothing
$$ are incompatible.
\end{proof}

\begin{lem}\label{Lemma4}
If $\alpha $ is a cyclically-interval $t$-coloring of a tree $H$,
$t\in\Theta(H)$, $V(H)=\{b_{1},...,b_{p}\}$, $p\geq 1$, then there
are vertices $b'\in V(H),$ $b''\in V(H)$ such that $[1,t]=\alpha
[TP(b',b'')].$
\end{lem}

\begin{proof}
Assume the contrary. Suppose that for an arbitrary $b_{i}\in V(H),$
$b_{j}\in V(H),$ $\alpha [TP(b_{i},b_{j})]\subset[1,t]$. Set:
$\max\Big\{\big| \alpha[TP(b_{i},b_{j})]\big|/1\leq i\leq p, 1\leq
j\leq p\Big\}\equiv m_0.$ It is clear that $m_{0}<t$. Without loss
of generality we may assume that $m_{0}\geq 2$. Consider the simple
path $P_{0}=(x_{0},e_{1},x_{1},...,x_{k-1},e_{k},x_{k})$ of the tree
$H$ with $\big| \alpha [TP_{0}]\big| =m_{0}$. Clearly, without loss
of generality, we may assume that $k\geq 2$.

Lemma \ref{Lemma2} implies that there are $i^{\prime }\in[1,t]$,
$i^{\prime \prime}\in[1,t]$, and $j^{\prime }\in \{1,2\}$, for which
$\alpha\Bigg[\bigcup\limits_{i=1}^{k-1}J_H(x_{i})\Bigg]=intcyc_{j^{\prime
}}[(i^{\prime },i^{\prime \prime }),t]$. As $m_{0}<t$, there is
$\tau _{0}\in [1,t]$ such that $\tau _{0}\not\in intcyc_{j^{\prime
}}[(i^{\prime },i^{\prime \prime }),t]$.

Consider an edge $e^{1}\in E(H)$ for which $\alpha(e^{1})=\tau
_{0}$, and assume that $e^{1}=(u_{0},u_{1})$. Clearly, $e^{1}\not\in
TP_{0}(x_{0},x_{k})$.

Without loss of generality we may assume that
$\rho_{H}(u_{1},\tilde{V}P_{0}(x_{0},x_{k}))<\\
\rho_{H}(u_{0},\tilde{V}P_{0}(x_{0},x_{k}))$. Let $z_{0}\in
\tilde{V}P_{0}(x_{0},x_{k})$ be the vertex with $\rho
_{H}(u_{1},z_{0})=\rho _{H}(u_{1},\tilde{V}P_{0}(x_{0},x_{k}))$. It
is not hard to see that
$z_{0}\in\tilde{V}P_{0}(x_{0},x_{k})\backslash
intVP_{0}(x_{0},x_{k})$ and for any $z^{\prime }\in
\tilde{V}P_{0}(x_{0},x_{k})\backslash intVP_{0}(x_{0},x_{k}),$
$z^{\prime }\neq z_{0}$, $\rho _{H}(u_{1},z_{0})<\rho
_{H}(u_{1},z^{\prime })$.

\case{1} $z_{0}=x_{0}$. Clearly, $\big| \alpha
[TP(u_{0},x_{k})]\big| \geq m_{0}+1$, which contradicts the choice
of $P_{0}$.

\case{2} $z_{0}=x_{k}$. This case is considered similarly as the
case 1.

\case{3} $z_{0}\neq x_{0}$, $z_{0}\neq x_{k}$.

Clearly, there is $\tilde{x}\in intVP_{0}(x_{0},x_{k})$ such that
$z_{0}\in I_H(\tilde{x})$. Suppose that $\alpha
((z_{0},\tilde{x}))=\tau ^{\prime } $. Clearly, $i^{\prime }\neq
i^{\prime \prime }$.

\case{3a} $\tau ^{\prime }=i^{\prime }$.

Lemma \ref{Lemma3}, the equalities $\alpha (e^{1})=\tau _{0}$,
$\alpha ((z_{0},\tilde{x}))=i^{\prime }$, and the definition of the
path $P(u_{0},\tilde{x})$ imply that $\exists j_{1}\in \{1,2\}$ such
that $intcyc_{j_{1}}[(\tau_{0},i^{\prime }),t]\subseteq \alpha
\Bigg[\bigcup\limits_{x\in intVP(u_{0}, \tilde{x})}J_H(x)\Bigg]$.
Consider the edge $\tilde{e}\in TP_{0}(x_{0},x_{k})$ with $\alpha
(\tilde{e})=i^{\prime \prime }$. Assume: $\tilde{e}=(x^{\prime
},x^{\prime \prime })$. Without loss of generality we may assume
that $\rho _{H}(z_{0},x^{\prime })<\rho _{H}(z_{0},x^{\prime \prime
})$. It is not hard to check that $TP(z_{0},x^{\prime \prime
})\subseteq TP_{0}(x_{0},x_{k})$, therefore, by the choice of $\tau
_{0}$, we have $\tau _{0}\not\in \alpha [TP(z_{0},x^{\prime \prime
})]$. Lemma \ref{Lemma2} implies that $\alpha [TP(z_{0},x^{\prime
\prime })]$ is a $t$-cyclic interval.

Clearly, $\exists j_{2}\in \{1,2\}$ such that $\tau _{0}\in
intcyc_{j_{2}}((i^{\prime },i^{\prime \prime }),t)$, and, therefore,\\
$intcyc_{j_{2}}((i^{\prime },i^{\prime \prime }),t)\nsubseteq \alpha
[TP(z_{0},x^{\prime \prime })]$.

This conclusion, the equalities $\alpha((z_{0},\tilde{x}))=i^{\prime
}$, $\alpha (\tilde{e})=i^{\prime \prime }$, and Lemma \ref{Lemma3}
imply that $ intcyc_{3-j_{2}}[(i^{\prime },i^{\prime \prime
}),t]\subseteq \alpha [TP(z_{0},x^{\prime \prime })]$, hence $\big|
\alpha [TP(u_{0},x^{\prime \prime })]\big| \geq m_{0}+1$, which
contradicts the choice of $P_{0}$.

\case{3b} $\tau ^{\prime }=i^{\prime \prime }$. This case is
considered similarly as the case 3a with interchanging of the roles
of $i^{\prime}$ and $i^{\prime \prime }$.

\case{3c} $\tau ^{\prime }\not\in \{i^{\prime },i^{\prime
\prime}\}$.

Lemma \ref{Lemma3}, the equalities $\alpha (e^{1})=\tau _{0}$,
$\alpha ((z_{0},\tilde{x}))=\tau ^{\prime }$, and the definition of
the path $P(u_{0},\tilde{x})$ imply that $\exists j_{1}\in \{1,2\}$
such that $intcyc_{j_{1}}[(\tau _{0},\tau ^{\prime }),t]\subseteq
\alpha \Bigg[\bigcup\limits_{x\in
intVP(u_{0},\tilde{x})}J_H(x)\Bigg]$. This implies that at least one
of the following statements is true:
\begin{enumerate}
\item $i^{\prime }\in intcyc_{j_{1}}[(\tau _{0},\tau ^{\prime }),t],$
\item $i^{\prime \prime }\in intcyc_{j_{1}}[(\tau _{0},\tau ^{\prime }),t]$.
\end{enumerate}

Without loss of generality let us assume that the statement 1) is
true. Consider the edge $\tilde{e}\in TP_{0}(x_{0},x_{k})$ with
$\alpha (\tilde{e})=i^{\prime \prime }$. Assume:
$\tilde{e}=(x^{\prime },x^{\prime \prime })$. Without loss of
generality we may assume that $\rho _{H}(z_{0},x^{\prime })<\rho
_{H}(z_{0},x^{\prime \prime })$. It is not hard to check that
$TP(z_{0},x^{\prime \prime })\subseteq TP_{0}(x_{0},x_{k})$,
therefore, by the choice of $\tau _{0}$, we have $\tau _{0}\not\in
\alpha [TP(z_{0},x^{\prime \prime })]$. Lemma \ref{Lemma2} implies
that $\alpha [TP(z_{0},x^{\prime \prime })]$ is a $t$-cyclic
interval.

Clearly, $\exists j_{2}\in \{1,2\}$ such that $\tau _{0}\in
intcyc_{j_{2}}((\tau ^{\prime },i^{\prime \prime }),t)$, and,
therefore, \\ $intcyc_{j_{2}}((\tau ^{\prime },i^{\prime \prime
}),t)\nsubseteq \alpha [TP(z_{0},x^{\prime \prime })]$. This
conclusion, the equalities $\alpha ((z_{0},\tilde{x}))=\tau ^{\prime
}$, $\alpha (\tilde{e})=i^{\prime \prime }$, and Lemma \ref{Lemma3}
imply that $intcyc_{3-j_{2}}[(\tau ^{\prime },i^{\prime \prime
}),t]\subseteq \alpha [TP(z_{0},x^{\prime \prime })]$, hence $\big|
\alpha [TP(u_{0},x^{\prime \prime })]\big| \geq m_{0}+1$, which
contradicts the choice of $P_{0}$.
\end{proof}

\begin{cor}\label{Corollary1}
If $\alpha $ is a cyclically-interval $t$-coloring of a tree $H$,
where $t\in\Theta(H)$, then there are vertices $x^{\prime }\in
V(H)$, $x^{\prime \prime }\in V(H)$ such that $t\leq \left\vert
TP(x^{\prime },x^{\prime \prime })\right\vert $.
\end{cor}

\begin{proof}
Since the inequality $\big| \alpha [TP(x,y)]\big| \leq \left\vert
TP(x,y)\right\vert $ holds for arbitrary vertices $x\in V(H)$, $y\in
V(H),$ it is not difficult to notice that our statement follows from
Lemma \ref{Lemma4}.
\end{proof}

\begin{cor}
If $\alpha $ is a cyclically-interval $W_{cyc}(H)$-coloring of a
tree $H$, then there are vertices $x^{\prime }\in V(H)$,
$x^{\prime\prime}\in V(H)$ such that $W_{cyc}(H)\leq\left\vert
TP(x^{\prime},x^{\prime\prime})\right\vert$.
\end{cor}

\begin{cor}\label{Cor4}
For any tree $H$, $W_{cyc}(H)\leq M(H)$.
\end{cor}

\begin{thm}
For any tree $H$, $W_{cyc}(H)=M(H)$.
\end{thm}
\textit{Proof} follows from Corollaries \ref{Cor1} and \ref{Cor4}.


\begin{cor}\cite{Shved1_11}
Let $H$ be an arbitrary tree. Then
\begin{enumerate}
\item $H\in \mathfrak{M}$,
\item $w_{cyc}(H)=\Delta (H),$
\item $W_{cyc}(H)=M(H),$
\item $\Theta(H)=[\Delta(H),M(H)]$.
\end{enumerate}
\end{cor}

\begin{cor}
For an arbitrary tree $H$ and any positive integer $t$, $H\in
\mathfrak{M}_t$ if and only if $H\in \mathfrak{N}_t.$
\end{cor}

\section{Acknowledgment}

The author thanks the anonymous reviewer for his useful advices and
suggestions.

\end{document}